\documentclass[letterpaper, 10 pt, conference]{ieeeconf}%
\usepackage{bm}
\usepackage{amssymb}
\usepackage{epsfig}
\usepackage{subfigure}
\usepackage{amsmath}
\usepackage{xcolor}
\usepackage{graphicx}
\usepackage{epstopdf}
\usepackage{amsfonts}%
\usepackage{indentfirst}
\newtheorem{problem}{\textbf{Problem}}
\newtheorem{definition}{\textbf{Definition}}

\newtheorem{theorem}{\rm\textbf{Theorem}}

\newtheorem{remark}{\rm\textbf{Remark}}
\providecommand{\U}[1]{\protect\rule{.1in}{.1in}}
\IEEEoverridecommandlockouts
\begin{document}

\title{{\LARGE \textbf{Adaptive Control Barrier Functions for Safety-Critical Systems}}}
\author{Wei Xiao, Calin Belta and Christos G. Cassandras\thanks{This work was supported in part by
		NSF under grants IIS-1723995, 
		CPS-1446151, ECCS-1509084, DMS-1664644, CNS-1645681, by AFOSR under grant
		FA9550-19-1-0158, by ARPA-E's NEXTCAR program under grant DE-AR0000796 and by
		the MathWorks.}\thanks{The authors are
with the Division of Systems Engineering and Center for Information and
Systems Engineering, Boston University, Brookline, MA, 02446, USA
\texttt{{\small \{xiaowei,cbelta,cgc\}@bu.edu}}}}
\maketitle

\begin{abstract}
Recent work showed that stabilizing affine control systems to desired (sets of) states while optimizing quadratic costs  and observing state and control constraints can be reduced to quadratic programs (QP) by using control barrier functions (CBF) and control Lyapunov functions. In our own recent work, we defined high order CBFs (HOCBFs) to accommodating systems and constraints with arbitrary relative degrees, and a penalty method to increase the feasibility of the corresponding QPs. In this paper, we introduce adaptive CBF (AdaCBFs) that can accommodate time-varying control bounds and dynamics
noise, and also address the feasibility problem. Central to our approach is the introduction of penalty functions in the definition of an AdaCBF and the definition of auxiliary dynamics for these penalty functions that are HOCBFs and are stabilized by CLFs. We demonstrate the advantages of the proposed method by applying it to a cruise control problem with different road surfaces, tires slipping, and dynamics noise. 
\end{abstract}

\thispagestyle{empty} \pagestyle{empty}


\section{INTRODUCTION}
\label{sec:intro}

Barrier functions (BFs) are
Lyapunov-like functions \cite{Tee2009}\cite{Wieland2007}, whose use can be traced back to optimization problems \cite{Boyd2004}. More recently, they have been employed to prove set invariance \cite{Aubin2009}\cite{Prajna2007}\cite{Wisniewski2013} and for multi-objective control \cite{Panagou2013}. In \cite{Tee2009}, it was proved that if a BF for a given set satisfies Lyapunov-like conditions, then the set is forward invariant. A less restrictive form of a BF, which is allowed to grow when far away from the boundary of the set, was proposed in \cite{Aaron2014}.
Another approach that allows a BF to be zero was proposed in \cite{Glotfelter2017} \cite{Lindemann2018}. This simpler form has also been considered in time-varying cases and applied to enforce Signal Temporal Logic (STL) formulas as hard constraints \cite{Lindemann2018}.

Control BFs (CBFs) are extensions of BFs for control systems, and are used to map a constraint that is defined over system states to a constraint on the control input. Recently, it has been shown that, to stabilize an affine control system while optimizing a quadratic cost and satisfying states and control constraints, CBFs can be combined with
control Lyapunov functions (CLFs) \cite{Sontag1983}\cite{Artstein1983}\cite{Freeman1996}\cite{Aaron2012} to form quadratic programs (QPs) \cite{Galloway2013}\cite{Aaron2014}\cite{Glotfelter2017} that are solved in real time.

The CBFs from \cite{Aaron2014} and \cite{Glotfelter2017} work for constraints that have relative degree one with respect to the system dynamics.
A backstepping approach was introduced in \cite{Hsu2015} to address higher relative degree constraints, and it was shown to work for relative degree two. A CBF method for position-based constraints with relative degree two was also proposed in \cite{Wu2015}. A more general form, which works
for arbitrarily high relative degree constraints, was proposed in \cite{Nguyen2016}. The method in \cite{Nguyen2016} employs input-output linearization and finds a pole placement controller with negative poles to stabilize the CBF to zero. Thus, this CBF is an exponential CBF. The high order CBF (HOCBF) that we proposed in \cite{Xiao2019} is simpler and more general than the exponential CBF \cite{Nguyen2016}.
However, the QPs can easily be infeasible when both state constraints (enforced by HOCBFs) and tight control bounds are involved. Although the
penalties involved in the definition of the HOCBF can help to improve feasibility \cite{Xiao2019}, this might not work under time-varying control bounds and dynamics noise. In addition, the HOCBF method is conservative in the sense that
the satisfaction of the HOCBF constraint is only sufficient for the satisfaction of the original constraint, which can limit the system performance.

To improve the problem feasibility under time-varying control bounds and dynamics noise, in this paper we propose adaptive CBFs (AdaCBFs). The proposed AdaCBFs can also help to alleviate the conservativeness of the HOCBF method. Specifically, we introduce penalty functions in the definition of an AdaCBF, and define auxiliary dynamics for these penalty functions that are HOCBFs (such that they are guaranteed to be non-negative) and are stabilized by CLFs. This way, the AdaCBF constraint is relaxed by the penalty functions through the control inputs of the auxiliary dynamics, while the forward invariance property of the HOCBF method is guaranteed. Since the AdaCBF constraint is relaxed through the penalty functions, we show that its satisfaction is a necessary and sufficient condition for the satisfaction of the original constraint, which leads to improvements in the performance of the system.

We formulate optimal control problems with constraints given by AdaCBFs and CLFs, and show the adaptivity of the proposed AdaCBF on an adaptive cruise control (ACC) problem with different and time-varying control bounds (e.g., on different road surfaces and with tires slipping), as well as with dynamics noise. The results clearly demonstrate the advantages of the proposed AdaCBF.

We give preliminaries on HOCBF and CLF in Sec. \ref{sec:pre} before we introduce the AdaCBF in Sec. \ref{sec:hocbf}. We first formulate and then reformulate the ACC problem using AdaCBF in Sec. \ref{sec:ACC} and Sec. \ref{sec:reform}, respectively. Case studies are presented in Sec. \ref{sec:case}, followed by conclusions and final remarks in Sec. \ref{sec:conclusion}.

 
 \section{PRELIMINARIES}
\label{sec:pre}
\begin{definition} \label{def:classk}
	({\it Class $\mathcal{K}$ function} \cite{Khalil2002}) A continuous function $\alpha:[0,a)\rightarrow[0,\infty), a > 0$ is said to belong to class $\mathcal{K}$ if it is strictly increasing and $\alpha(0)=0$. 
\end{definition}

Consider an affine control system of the form
\begin{equation} \label{eqn:affine}
\dot {\bm{x}} = f(\bm x) + g(\bm x)\bm u
\end{equation}
where  $\bm x\in \mathbb{R}^n$, $f:\mathbb{R}^n\rightarrow \mathbb{R}^n$ and $g:\mathbb{R}^n \rightarrow \mathbb{R}^{n\times q}$ are globally Lipschitz, and $\bm u\in U \subset \mathbb{R}^q$ ($U$ denotes the control constraint set). Solutions $\bm x(t)$ of (\ref{eqn:affine}), starting at $\bm x(0)$ (we set the initial time to 0 without loss of generality), $t\geq 0$, are forward complete.

Suppose the control bound $U$ is defined as (the inequality is interpreted componentwise):
\begin{equation} \label{eqn:control}
U:= \{\bm u\in\mathbb{R}^q: \bm u_{min}\leq\bm u \leq \bm u_{max}\},
\end{equation}
with $\bm u_{min},\bm u_{max}\in \mathbb{R}^q$.

\begin{definition} \label{def:forwardinv}
	A set $C\subset\mathbb{R}^n$ is forward invariant for system (\ref{eqn:affine}) if its solutions starting at any $\bm x(0) \in C$ satisfy $\bm x(t)\in C$ for $\forall t\geq 0$.
\end{definition} 

\begin{definition} \label{def:relative}
	({\it Relative degree})
	The relative degree of a (sufficiently many times) differentiable function $b:\mathbb{R}^n\rightarrow \mathbb{R}$ with respect to system (\ref{eqn:affine}) is the number of times we need to differentiate it along its dynamics until the control $\bm u$ explicitly shows in the corresponding derivative. 
\end{definition}

In this paper, since function $b$ is used to define a constraint $b(\bm x)\geq 0$, we will also refer to the relative degree of $b$ as the relative degree of the constraint.

For a constraint $b(\bm x)\geq 0$ with relative degree $m$, $b: \mathbb{R}^n \rightarrow \mathbb{R}$, and $\psi_0(\bm x) := b(\bm x)$, we define a sequence of functions  $\psi_i: \mathbb{R}^n \rightarrow \mathbb{R}, i\in\{1,2,\dots,m\}$:
\begin{equation} \label{eqn:functions}
\begin{aligned}
\psi_i(\bm x) := \dot \psi_{i-1}(\bm x) + \alpha_i(\psi_{i-1}(\bm x)),i\in\{1,2,\dots,m\},
\end{aligned}
\end{equation}
where $\alpha_i(\cdot),i\in\{1,2,\dots,m\}$ denotes a $(m-i)^{th}$ order differentiable class $\mathcal{K}$ function.

We further define a sequence of sets $C_i, i\in\{1,2,\dots,m\}$ associated with (\ref{eqn:functions}) in the form:
\begin{equation} \label{eqn:sets}
\begin{aligned}
C_i := \{\bm x \in \mathbb{R}^n: \psi_{i-1}(\bm x) \geq 0\}, i\in\{1,2,\dots,m\}.
\end{aligned}
\end{equation}

\vspace{2ex}
\begin{definition} \label{def:hocbf}
	({\it High Order Control Barrier Function (HOCBF)} \cite{Xiao2019}) Let $C_1, C_2,\dots, C_{m}$ be defined by (\ref{eqn:sets}) and $\psi_1(\bm x), \psi_2(\bm x),\dots, \psi_{m}(\bm x)$ be defined by (\ref{eqn:functions}). A function $b: \mathbb{R}^n\rightarrow \mathbb{R}$ is a high order control barrier function (HOCBF) of relative degree $m$ for system (\ref{eqn:affine}) if there exist $(m-i)^{th}$ order differentiable class $\mathcal{K}$ functions $\alpha_i,i\in\{1,2,\dots,m-1\}$  and a class $\mathcal{K}$ function $\alpha_{m}$ such that
	\begin{equation}\label{eqn:constraint}
	\begin{aligned}
	L_f^{m}b(\bm x) + L_gL_f^{m-1}b(\bm x)\bm u + S(b(\bm x)) + \alpha_m(\psi_{m-1}(\bm x)) \geq 0,
	\end{aligned}
	\end{equation}
	for all $\bm x\in C_1 \cap C_2\cap,\dots, \cap C_{m}$. In (\ref{eqn:constraint}), $L_f^m$ ($L_g$) denotes Lie derivatives along $f$ ($g$) $m$ (one) times, $S(\cdot)$ denotes the remaining Lie derivatives along $f$ with degree less than or equal to $m-1$ (omitted for simplicity, see \cite{Xiao2019}).
\end{definition}

The HOCBF is a general form of the relative degree one CBF \cite{Aaron2014} \cite{Glotfelter2017} \cite{Lindemann2018} (set $m = 1$ in the HOCBF, and we can show that the CBF form in \cite{Aaron2014} is the same as \cite{Glotfelter2017} \cite{Lindemann2018}), and it is also a general form of the exponential CBF \cite{Nguyen2016} for high relative degree constraints (define all the class $\mathcal{K}$ functions in linear form in Def. \ref{def:hocbf}).

\vspace{2ex}
Given a HOCBF $b$, we define the set of all control values that satisfy (\ref{eqn:constraint}) as:
\begin{equation}
\begin{aligned}
K_{cbf} = \{\bm u\in U: L_f^{m}b(\bm x) + L_gL_f^{m-1}b(\bm x)\bm u \\ + S(b(\bm x))  + \alpha_m(\psi_{m-1}(\bm x)) \geq 0\}
\end{aligned}
\end{equation}

\vspace{2ex}
\begin{theorem} \label{thm:hocbf}
	(\cite{Xiao2019}) Given a HOCBF $b(\bm x)$ from Def. \ref{def:hocbf} with the associated sets $C_1, C_2,\dots, C_{m}$ defined by (\ref{eqn:sets}), if $\bm x(0) \in C_1 \cap C_2\cap,\dots,\cap C_{m}$, then any Lipschitz continuous controller $\bm u(t)\in K_{cbf}, \forall t\geq 0$ renders
	$C_1$ $\cap  C_2\cap,\dots, \cap C_{m}$ forward invariant for system (\ref{eqn:affine}).
\end{theorem}

\begin{definition}  \label{def:clf}
	({\it Control Lyapunov function (CLF)} \cite{Aaron2012}) A continuously differentiable function $V: \mathbb{R}^n\rightarrow \mathbb{R}$ is a globally and exponentially stabilizing control Lyapunov function (CLF) for system (\ref{eqn:affine}) if there exist constants $c_1 >0, c_2>0, c_3>0$ such that
	\begin{equation}
	c_1||\bm x||^2 \leq V(\bm x) \leq c_2 ||\bm x||^2
	\end{equation}
	\begin{equation}\label{eqn:clf}
	\underset{u\in U}{inf} \lbrack L_fV(\bm x)+L_gV(\bm x) \bm u + c_3V(\bm x)\rbrack \leq 0.
	\end{equation}
	for $\forall \bm x\in \mathbb{R}^n$.
\end{definition}

\begin{theorem} \label{thm:clf}
	(\cite{Aaron2012}) Given a CLF $V$ as in Def. \ref{def:clf}, any Lipschitz continuous controller $ \bm u \in K_{clf}(\bm x)$, with
	$$K_{clf}(\bm x) := \{\bm u\in U: L_fV(\bm x)+L_gV(\bm x) \bm u + c_3V(\bm x) \leq 0\},$$
	exponentially stabilizes system (\ref{eqn:affine}) to its zero dynamics (defined by the dynamics of the internal part if we transform the system to standard form and set the output to zero \cite{Khalil2002}). 
\end{theorem}

Note that (\ref{eqn:clf}) can be relaxed by replacing 0 with a relaxation variable at its right-hand side, and we wish to minimize this relaxation variable \cite{Aaron2012}. Recent works \cite{Aaron2014}, \cite{Lindemann2018}, \cite{Nguyen2016} combine CBFs and CLFs with quadratic costs to form optimization problems. Time is discretized and an optimization problem with constraints given by CBFs and CLFs is solved at each time step. Note that these constraints are linear in control since the state is 
fixed at the value at the beginning of the interval, and therefore the optimization problem is a quadratic program (QP). The optimal control obtained by solving the QP is applied 
at the current time step and held constant for the whole interval. The dynamics (\ref{eqn:affine}) are updated, and the procedure is repeated. This method works conditioned on the fact that the QP at every time step is feasible. However, this is not guaranteed, in particular under tight control bounds. In this paper, we show how the QP feasibility can be improved by using adaptive CBFs.

\section{ADAPTIVE CONTROL BARRIER FUNCTIONS}
\label{sec:hocbf}

In this section, we define adaptive CBFs (AdaCBFs). We use a simple example to motivate the need for such functions and to illustrate the main ideas. 

\subsection{Example: Simplified Adaptive Cruise Control}
\label{sec:hocbf:moti}

Consider the simplified adaptive cruise control (SACC) problem\footnote{A more realistic version of this problem, called the adaptive cruise control problem (ACC), is defined in Sec.  \ref{sec:ACC}.} with the ego (controlled) vehicle dynamics in the form:
\begin{equation} \label{eqn:simpledynamics}
\left[\begin{array}{c} 
\dot x(t)\\
\dot v(t)
\end{array} \right]=
\left[\begin{array}{c}  
v(t)\\
0
\end{array} \right] + 
\left[\begin{array}{c}  
0\\
1
\end{array} \right]u(t),
\end{equation}
where $x(t)$ and $v(t)$ denote the position and velocity of the controlled vehicle along its lane, respectively, and $u(t)$ is its control.

We also have control constraints:
\begin{equation}\label{eqn:scons}
u_{min} \leq u(t) \leq u_{max}, \forall t\geq 0,
\end{equation}
where $u_{min}<0$ and $u_{max} > 0$ are the minimum and maximum control input, respectively.

We require that the distance between the ego vehicle and its immediately preceding vehicle
(the coordinates $x(t)$ of the ego vehicle and $x_{p}(t)$ of the preceding vehicle, respectively, are measured from the same origin and $x_{p}(t) \geq x(t), \forall t\geq 0$)  
be greater than $\delta_0 > 0$, i.e.,
\begin{equation} \label{eqn:safety}
x_{p}(t) - x(t) \geq \delta_0,\forall t\geq 0.
\end{equation}

Assume the preceding vehicle runs at constant speed $v_0$. Let $\bm x(t) := (x(t),v(t))$ and $b(\bm x(t)) := x_p(t) - x(t) - \delta_0$. The relative degree of $b(\bm x(t))$ is 2, so we choose a HOCBF with $m = 2$. We define $\psi_0(\bm x(t)):= b(\bm x(t))$, $\alpha_1(\psi_0(\bm x(t))):=\psi_0(\bm x(t))$ and $\alpha_2(\psi_1(\bm x(t))) := \psi_1(\bm x(t))$ in Def. \ref{def:hocbf} and find a control for the ego vehicle such that the constraint (\ref{eqn:safety}) is satisfied. The control $u(t)$ should satisfy:
\begin{equation}\label{eqn:safety_ex2}
\begin{aligned}
 \underbrace{0}_{L_f^2b(\bm x(t))} + \underbrace{-1}_{L_gL_fb(\bm x(t))}\times u(t) + \underbrace{v_0 - v(t)}_{S(b(\bm x(t)))} \\+  \underbrace{v_0 - v(t) +x_p(t) - x(t) - \delta_0}_{\alpha_2(\psi_1(\bm x(t)))} \geq 0.
\end{aligned}
\end{equation}

Suppose we wish to minimize $\int_0^T u^2(t)$. We can then use the QP based method introduced at the end of the last section to solve this SACC problem. However, the HOCBF constraint (\ref{eqn:safety_ex2}) can easily conflict with $u_{min}$ in (\ref{eqn:scons}) when the two vehicles get close to each other, as shown in \cite{Xiao2019}. When this happens, the QP will be infeasible. We can use the penalty method from \cite{Xiao2019} to improve the QP feasibility, i.e., by adding a positive constant penalty $p \in \mathbb{R}$ on both $\alpha_1(\cdot), \alpha_2(\cdot)$. Then the control $u(t)$ should satisfy:
\begin{equation}\label{eqn:safety_pe}
\begin{aligned}
\underbrace{0}_{L_f^2b(\bm x(t))} + \underbrace{-1}_{L_gL_fb(\bm x(t))}\times u(t) + \underbrace{p(v_0 - v(t))}_{S(b(\bm x(t)))} \\+  \underbrace{p(v_0 - v(t)) +p^2(x_p(t) - x(t) - \delta_0)}_{p\alpha_2(\psi_1(\bm x(t)))} \geq 0.
\end{aligned}
\end{equation}

Given $u_{min}$, we can find a small enough value for $p$ such that (\ref{eqn:safety_pe}) will not conflict with $u_{min}$ in (\ref{eqn:scons}), i.e., the QP is always feasible. However, the lower bound of the control $u_{min}$ is not a constant in reality. The value of $u_{min}$ depends on the weather condition and the road surface roughness, etc.. For example, the vehicle maximum braking force (directly corresponding to $u_{min}$) in a rainy day is usually smaller than the one in a sunny day. When we choose a proper $p$ for the HOCBF constraint (\ref{eqn:safety_pe}) in the sunny day such that the QP in the SACC problem is always feasible, the QP may be infeasible in the rainy day. The unknown roughness of the road surface can further make the choice of the $p$ value difficult. Moreover, the assumption of a constant speed $v_0$ for the front vehicle is too strong, and there are also vehicle dynamics noise that could make the QP infeasible. This motivates us to define an AdaCBF that works for time-varying control bounds and dynamics noise (i.e., the QP is always feasible).

\subsection{Adaptive Control Barrier Function (AdaCBF)}
\label{sec:hobf}

We consider a function that defines an invariant set for system (\ref{eqn:affine}). For a relative degree $m$ function $b: \mathbb{R}^n \rightarrow \mathbb{R}$, let $\psi_0(\bm x) := b(\bm x)$. Instead of using a constant penalty $p_i > 0, i\in \{1,2,\dots,m\}$ for each class $\mathcal{K}$ function $\alpha_i(\cdot)$ for the penalty method \cite{Xiao2019} in the definition of a HOCBF, we define $p_i, i\in \{1,2,\dots,m\}$ as a \textbf{penalty function} of time $p_i(t)\geq 0$ and multiply it to each class $\mathcal{K}$ function $\alpha_i(\cdot)$, respectively. Let $\bm p(t):=(p_1(t), p_2(t),\dots, p_m(t))$. Thus, we define a sequence of functions $\psi_i: \mathbb{R}^n\times \mathbb{R}^m \rightarrow \mathbb{R}, i\in \{1,2,\dots,m\}$ in the form:
\begin{equation} \label{eqn:adafunc}
\begin{aligned}
\psi_1(\bm x, \bm p(t)) &:= \dot \psi_{0}(\bm x) + p_1(t)\alpha_1(\psi_{0}(\bm x)),\\
\psi_i(\bm x, \bm p(t)) &:= \dot \psi_{i-1}(\bm x, \bm p(t)) + p_i(t)\alpha_i(\psi_{i-1}(\bm x,\bm p(t))),\\ &\qquad \qquad \qquad\qquad\qquad\quad\;\; i\in \{2,\dots,m\},
\end{aligned}
\end{equation}
where $\alpha_i(\cdot), i\in \{1,2,\dots,m-1\}$ is a $(m-i)^{th}$ order differentiable class $\mathcal{K}$ function, $\alpha_{m}(\cdot)$ is a class $\mathcal{K}$ function.

 We further define a sequence of sets $C_i, i\in \{1,2,\dots,m\}$ associated with (\ref{eqn:adafunc}) in the form:
\begin{equation} \label{eqn:adasets}
\begin{aligned}
C_1 &:= \{\bm x \in \mathbb{R}^n: \psi_0(\bm x) \geq 0\},\\
C_i &:= \{(\bm x, \bm p(t)) \in \mathbb{R}^n\times \mathbb{R}^{m}: \psi_{i-1}(\bm x,\bm p(t)) \geq 0\},
\end{aligned}
\end{equation}
where $i\in \{2,\dots,m\}$.

The remaining question is how to choose $p_i(t), i\in \{1,2,\dots,m\}$. As shown in the definition of a HOCBF, $p_i(t), i\in \{1,2,\dots,m-1\}$ will be differentiated $m-i$ times, while $p_{m}(t)$ with not be differentiated (so we can just set $p_{m}(t)$ as a variable). Let $\bm p_i(t) := (p_i(t), p_{i,2}(t), \dots, p_{i,m-i}(t)) \in\mathbb{R}^{m-i} , i\in \{1,2,\dots,m-2\}$ where $p_{i,j} \in \mathbb{R}, j\in\{2,3,\dots,m-i\}$ denotes a new defined state variable that will be used later ($\bm p_{m-1}(t) := p_{m-1}(t)\in\mathbb{R}$ is a single variable, and we need to differentiate it once). We define input-output linearizable auxiliary dynamics (take $p_i(t)$ as the output $y_i(t)\in\mathbb{R}$ with relative degree $m-i$) for each $p_i$ (we skip the time variable $t$ from now on for simplicity) in the form:
\begin{equation} \label{eqn:assist}
\begin{aligned}
\dot {\bm{p}}_i &= F_i(\bm p_i) + G_i(\bm p_i) \nu_i, i\in \{1,2,\dots,m-1\},\\
y_i &= p_i,
\end{aligned}
\end{equation}
where $F_i:\mathbb{R}^{m-i}\rightarrow \mathbb{R}^{m-i}, G_i:\mathbb{R}^{m-i}\rightarrow \mathbb{R}^{m-i}$, $\nu_i \in \mathbb{R}$ denotes the control input for the auxiliary dynamics (\ref{eqn:assist}). Note that the auxiliary dynamics (\ref{eqn:assist}) are defined in general form. For simplicity, we can just define the auxiliary dynamics (\ref{eqn:assist}) in linear form. For example, we define $\dot p_{m-2} = p_{m-2,2}, \dot p_{m-2,2} = \nu_{m-2}$ (since we need to differentiate $p_{m-2}$ twice as shown in Def. \ref{def:hocbf}), and define $\dot p_{m-1} = \nu_{m-1}$ (since we need to differentiate $p_{m-1}$ once). Theoretically, we can initialize $\bm p_i(0)$ to any real number vector as long as $p_i(0) > 0$.

\begin{definition} \label{def:acbf}
	Let $C_i, i\in \{1,\dots,m\}$ be defined by (\ref{eqn:adasets}), $\psi_i(\bm x, \bm p), i\in \{1,\dots,m\}$ be defined by (\ref{eqn:adafunc}), and the auxiliary dynamics be defined by (\ref{eqn:assist}). A function $b: \mathbb{R}^n\rightarrow \mathbb{R}$ is an adaptive control barrier function (AdaCBF) with relative degree $m$ for (\ref{eqn:affine}) if every $p_i, \forall i\in\{1,2,\dots,m-1\}$ is a HOCBF with relative degree $m-i$ for the auxiliary dynamics (\ref{eqn:assist}), and there exist $(m-i)^{th}, i\in\{1,2,\dots,m-1\}$ order differentiable class $\mathcal{K}$ functions $\alpha_i$, a class $\mathcal{K}$ function $\alpha_{m}$ such that 
\begin{equation}\label{eqn:acbf}
\begin{aligned}
L_f^{m}b(\bm x) + L_gL_f^{m-1}b(\bm x)\bm u + \sum_{i=1}^{m-1}L_{G_i}L_{F_i}^{m-i-1}p_i\alpha_i(\psi_{i-1})\nu_i + \\\sum_{i=1}^{m-1}L_{F_i}^{m-i}p_i\alpha_i(\psi_{i-1})+ R(b(\bm x),\bm p) + p_m\alpha_m(\psi_{m-1}) \geq 0,
\end{aligned}
\end{equation}
for all $\bm x\in C_1, (\bm x, \bm p)\in C_2\cap,\dots, \cap C_{m}$, and all $p_{m} \geq 0$. In (\ref{eqn:acbf}), $R(b(\bm x), \bm p)$ denotes the remaining Lie derivative terms of $b(\bm x)$ (or $\bm p$) along $f$ (or $F_i, i\in\{1,2,\dots, m-1\}$) with degree less than $m$ (or $m-i$). These terms are skipped for simplicity but examples can be found in the revisited example or in the next section.
\end{definition} 

Let $\bm \nu := (\nu_1,\nu_2,\dots,\nu_{m-1})$, where $\nu_i, i\in\{1,2,\dots,m-1\}$ comes from the auxiliary dynamics (\ref{eqn:assist}). Note that $\bm \nu$ is constrained by the HOCBF constraints (defined in the following constraint set) for each $p_i\geq 0, i\in\{1,2,\dots,m-1\}$ since $p_i$ is a HOCBF with relative degree $m-i$ for (\ref{eqn:assist}). We define a constraint set $U_{cbf}$ for $\bm \nu$ in the form:
\begin{equation}\label{eqn:cbf}
\begin{aligned}
U_{cbf} = \{\bm \nu\in\mathbb{R}^{m-1}: L_{F_i}^{m-i}p_i + L_{G_i}L_{F_i}^{m-i-1}p_i\nu_i + S(\bm p_i) \\+ \alpha_{m-i}(\psi_{i,m-i-1}(\bm p_i)) \geq 0, \forall i\in\{1,2,\dots,m-1\}\},
\end{aligned}
\end{equation}
where $\psi_{i,m-i-1}(\cdot)$ is defined in (\ref{eqn:functions}) by replacing $b$ with $p_i$.

Given an AdaCBF $b(\bm x)$, we consider all control values $(\bm u, \bm \nu)\in U\times U_{cbf}$ that satisfy:
\begin{equation}
\begin{aligned}
K_{acbf}(\bm x, \bm p) = \{(\bm u, \bm \nu)\in U\times U_{cbf}:
 L_gL_f^{m-1}b(\bm x)\bm u \\+ \sum_{i=1}^{m-1}L_{G_i}L_{F_i}^{m-i-1}p_i\alpha_i(\psi_{i-1})\nu_i + \sum_{i=1}^{m-1}L_{F_i}^{m-i}p_i\alpha_i(\psi_{i-1})\\ + L_f^{m}b(\bm x) + R(b(\bm x),\bm p) + p_m\alpha_m(\psi_{m-1}) \geq 0,
\}.
\end{aligned}
\end{equation}

\vspace{2mm}
\begin{theorem} \label{thm:acbf}
	Given an AdaCBF $b(\bm x)$ from Def. \ref{def:acbf} with the associated sets $C_1, C_2,\dots, C_{m}$ defined by (\ref{eqn:adasets}), if $\bm x(0) \in C_1$ and $(\bm x(0), \bm p(0)) \in  C_2\cap,\dots,\cap C_{m}$, then any Lipschitz continuous controller $(\bm u(t), \bm \nu(t))\in K_{acbf}(\bm x(t), \bm p(t)), \forall t\geq 0$ renders the set $C_1$ forward invariant for system (\ref{eqn:affine}) and  $ C_2\cap,\dots, \cap C_{m}$ forward invariant for systems (\ref{eqn:affine}), (\ref{eqn:assist}).
\end{theorem}

\begin{proof}
 If $b(\bm x)$ is an AdaCBF, then we have that $p_{m}(t) \geq 0, \forall t\geq 0$. Constraint (\ref{eqn:acbf}) is equivalent to $\dot \psi_{m-1}(\bm x, \bm p) + p_m\alpha_{m}(\psi_{m-1}(\bm x, \bm p)) \geq 0$. If $p_m(t) > 0, \forall t\geq 0$, it follows from Thm. \ref{thm:hocbf} that $\psi_{m-1}(\bm x(t), \bm p(t))\geq 0, \forall t\geq 0$. If $p_m(t) = 0$, then $\dot \psi_{m-1}(\bm x, \bm p) \geq 0$. Since $(\bm x, \bm p) \in C_m$ (i.e., $\psi_{m-1}(\bm x, \bm p) \geq 0$ is initially satisfied), we have that $\psi_{m-1}(\bm x(t), \bm p(t))\geq 0, \forall t\geq 0$. Because $p_i, \forall i\in\{1,2,\dots,m-1\}$ is a HOCBF for the auxiliary dynamics (\ref{eqn:assist}), we have from Thm. \ref{thm:hocbf} that $p_i(t) \geq 0, \forall t\geq 0, \forall i\in\{1,2,\dots,m-1\}$. Then, we can recursively prove that $\psi_{i}(\bm x(t), \bm p(t))\geq 0, \forall t\geq 0, \forall i\in\{2,\dots,m-2\}$ similarly to the case $i = m-1$, and eventually prove that $\psi_0(\bm x(t)) \geq 0, \forall t\geq 0$, i.e., $b(\bm x(t)) \geq 0, \forall t\geq 0$.  Therefore, the sets $C_1$ and  $C_2\cap,\dots, \cap C_{m}$ are forward invariant.
\end{proof}

\begin{remark} \label{rem:adapt}
	 In the AdaCBF constraint (\ref{eqn:acbf}), the control input $\bm u$ of system (\ref{eqn:affine}) depends on the control input $\nu_i, \forall i\in\{1, 2,\dots,m-1\}$ of the auxiliary dynamics (\ref{eqn:assist}). The control input $\nu_i$ is only constrained by the corresponding HOCBF constraint (shown in (\ref{eqn:cbf})) since we require that $p_i$ is a HOCBF, and there are no control bounds on $\nu_i$. Therefore, we somewhat relax the constraints on the control input of system (\ref{eqn:affine}) in the AdaCBF by allowing the penalty function $p_i,\forall i\in\{1, 2,\dots,m\}$ to change. However, the forward invariance of the set $C_1$ is guaranteed, i.e., the original constraint $b(\bm x)\geq 0$ is guaranteed to be satisfied. This is the adaptivity of the AdaCBF. Note that we may not need to define a penalty function $p_i$ for every class $\mathcal{K}$ function $\alpha_i(\cdot)$ in (\ref{eqn:adafunc})-we can just define penalty functions for some of them.
\end{remark}

\textbf{Adaptivity to changing control bounds and dynamics noise:} In the HOCBF method, the problem may be infeasible in the presence of both control limitations (\ref{eqn:control}) and the HOCBF constraint (\ref{eqn:constraint}). There are two reasons for the problem to become infeasible: $(i)$ the control limitations (\ref{eqn:control}) are too tight or the control limitations are time-varying such that the HOCBF constraint (\ref{eqn:constraint}) will conlict with (\ref{eqn:control}) after it becomes active; $(ii)$ the dynamics (\ref{eqn:affine}) are not accurately modeled, and there may be uncertain variables, etc. (we consider all of them as dynamics noise). In this case, the HOCBF constraint (\ref{eqn:constraint}) might also conlict with (\ref{eqn:control}) when both of them become active. This is because the state variables also show up in the HOCBF constraint (\ref{eqn:constraint}), and thus, the dynamics noise can easily (and randomly) change the HOCBF constraint (\ref{eqn:constraint}) through the (noised) state variables such that (\ref{eqn:constraint}) may conflict with the control limitations when they are also active. However, the problem feasibility is improved in the AdaCBF method since the control $\bm u$ in the AdaCBF constraint (\ref{eqn:acbf}) is relaxed by $\nu_i,\forall i\in\{1, 2,\dots,m-1\}$, as discussed in Remark \ref{rem:adapt}. In order to make the original constraint $b(\bm x) \geq 0$ not be violated by the noise, we can define high order class $\mathcal{K}$ functions (such as high order polynomials) such that the value of the AdaCBF $b(\bm x)$ will stay away from 0 in the long run after the corresponding AdaCBF constraint (\ref{eqn:acbf}) becomes active, as shown in \cite{Xiao2019}.

\vspace{2mm}
\begin{theorem} \label{thm:nece}
	Given an AdaCBF $b(\bm x)$ from Def. \ref{def:acbf} with the associated sets $C_1, C_2,\dots, C_{m}$ defined by (\ref{eqn:adasets}), if $b(\bm x(0)) > 0$, then the satisfaction of the AdaCBF constraint (\ref{eqn:acbf}) is a necessary and sufficient condition for the satisfaction of the original constraint $b(\bm x)> 0$.
\end{theorem}

\begin{proof}
	If $b(\bm x(0)) > 0$, it follows from Thm. \ref{thm:acbf} that we can always choose proper class $\mathcal{K}$ functions (such as linear functions, quadratic functions, etc.) such that $\psi_0(\bm x(t)) > 0 $ and $\psi_i(\bm x(t),\bm p(t)) > 0, i\in\{1,\dots,m-1\}, \forall t\geq 0$ (note that $\psi_{0}(\bm x)=b(\bm x)$). Thus, the satisfaction of the AdaCBF constraint (\ref{eqn:acbf}) is a sufficient condition for the satisfaction of the original constraint $b(\bm x)> 0$. 
	
	If $b(\bm x(t)) > 0$, we have that there exists a penalty function $p_1(t) \geq 0$ (since $p_1(t)$ is a HOCBF) such that $\dot b(\bm x) > -p_1(t)\alpha_1(b(\bm x))$ for any $\dot b(\bm x)$ with respect to dynamics (\ref{eqn:affine}) (because $\alpha_1(b(\bm x)) > 0$). Note from (\ref{eqn:adafunc}) that $\psi_{0}(\bm x)=b(\bm x)$, we have $\dot \psi_0(\bm x) + p_1(t)\alpha_1(\psi_0(\bm x)) > 0$ (i.e., $\psi_1(\bm x, \bm p)> 0$). The $i^{th}$ derivative of $b(\bm x)$ shows in $\psi_i, i\in\{2,3,\dots,m-1\}$, and we can also prove similarly that there exists a penalty function $p_i(t)\geq 0$ (since $p_i(t)$ is a HOCBF) such that $\psi_i(\bm x, \bm p)> 0, i\in\{2,3,\dots,m-1\}$ in a recursive way. Eventually, there exists $p_m(t)\geq 0$ such that $\dot \psi_{m-1}(\bm x, \bm p) + p_m(t)\alpha_m(\psi_{m-1}(\bm x, \bm p)) \geq 0$ (i.e., $\psi_m(\bm x, \bm p)\geq 0$). Since $\psi_m(\bm x, \bm p)\geq 0$ is equivalent to the satisfaction of the AdaCBF constraint (\ref{eqn:acbf}), we have that the satisfaction of the AdaCBF constraint (\ref{eqn:acbf}) is a necessary condition for the satisfaction of the original constraint $b(\bm x)> 0$.
\end{proof}

\begin{remark}
	Since $p_i(t), i\in\{1,2,\dots,m\}$ is required to be $(m-i)^{th}$ order differentiable, it is usually not guaranteed to find such $p_i(t)$ functions when we find the control by solving the QPs introduced at the end of Sec. \ref{sec:pre}.  Since the satisfaction of the AdaCBF constraint (\ref{eqn:acbf}) is equivalent to the satisfaction of the original constraint, the system performance is not reduced in the mapping of a constraint from state to control.
\end{remark}

$\textbf{Example revisited.}$ For the SACC problem introduced in Sec.\ref{sec:hocbf:moti}, the relative degree of the constraint from Eqn. (\ref{eqn:safety}) is 2, i.e., we need an AdaCBF with $m = 2$. We  still choose  $\alpha_1(b(\bm x(t))) = b(\bm x(t))$ and $\alpha_2(\psi_1(\bm x(t))) = \psi_1(\bm x(t))$  in the definition of an AdaCBF in Def. \ref{def:acbf}. Suppose we only consider a penalty function $p_1$ on the class $\mathcal{K}$ function $\alpha_1(\cdot)$ and define linear dynamics for $p_1$ in the form $\dot p_1 = \nu_1$. In order for $b(\bm x(t)) := x_{p}(t) - x(t) - \delta_0$ to be an AdaCBF for (\ref{eqn:simpledynamics}), a control input $u(t)$ should satisfy
\begin{small}
\begin{equation}\label{eqn:safety_exre}
\begin{aligned}
\underbrace{0}_{L_f^2b(\bm x(t))} + \underbrace{-1}_{L_gL_fb(\bm x(t))}\times u(t) + \underbrace{(x_{p}(t) - x(t) - \delta_0)}_{L_{G_1}p_1(t)\alpha_1(b(\bm x(t)))}\times \nu_1(t) \\ + \underbrace{p_1(t)(v_0 - v(t))}_{R(b(\bm x(t)), \bm p(t))} + \underbrace{v_0 - v(t) + p_1(t)(x_{p}(t) - x(t) - \delta_0)}_{\alpha_2(\psi_1(\bm x(t),\bm p(t)))} \geq 0.
\end{aligned}
\end{equation}
\end{small}

Since $u(t)$ depends on $\nu_1$ that is without control bounds, the control input $u(t)$ in the above AdaCBF constraint is relaxed. Thus, this constraint is adaptive to the change of the control bound $u_{min}$ in (\ref{eqn:scons}) and the uncertainties of $v_0$ and $x_p(t)$ from the front vehicle, etc.. Note that $p_1$ should be a HOCBF for the auxiliary dynamics $\dot p_1 = \nu_1$. The control input $\nu_1$ is constrained by the corresponding HOCBF constraint such that $p_1(t)\geq 0, \forall t\geq 0$ is satisfied.

\subsection{Optimal Control with AdaCBFs}
\label{sec:oc}

 Consider an optimal control problem for system (\ref{eqn:affine}) with the cost defined as:
\begin{equation}\label{eqn:cost}
J(\bm u(t)) = \int_{0}^{T}\mathcal{C}(||\bm u(t)||)dt
\end{equation}
where $||\cdot||$ denotes the 2-norm of a vector,  $\mathcal{C}(\cdot)$ is a strictly increasing function of its argument, and $T > 0$. 

System (\ref{eqn:affine}) is not accurately modeled, as well as is with uncertain variables (both are unknown, and we can assign uncertain variables with the values of measured expection). In addition, system (\ref{eqn:affine}) has time-varying control bound $U(t)$ defined as:
\begin{equation} \label{eqn:controlvar}
U(t):= \{\bm u\in\mathbb{R}^q: \bm u_{min}(t)\leq\bm u \leq \bm u_{max}(t)\},
\end{equation}
with $\bm u_{min}(t), \bm u_{max}(t) \in \mathbb{R}^q$.

Assume a (safety) constraint $b(\bm x) \geq 0$ with relative degree $m$ has to be satisfied by system (\ref{eqn:affine}). In order to improve the problem feasibility, we use the AdaCBF method. Then $\bm u$ should satisfy the AdaCBF constraint (\ref{eqn:acbf}). Moreover, each $\nu_i, i\in\{1,2,\dots, m-1\}$ is constrained by the HOCBF constraint (\ref{eqn:constraint}) corresponding to the constraint $p_i(t)\geq 0$ for the auxiliary dynamics (\ref{eqn:assist}).

Note that the control $\bm \nu$ from the auxiliary dynamics are without control bounds, and are only constrained by the HOCBF constraints defined in (\ref{eqn:cbf}). However, the HOCBF constraint for each $\nu_i, i\in\{1,2,\dots,m-1\}$ is only constrained from one side, i.e.,
\begin{equation}
  \nu_i  \geq \frac{-L_{F_i}^{m-i}p_i -  S(\bm p_i) - \alpha_{m-i}(\psi_{i,m-i-1}(\bm p_i))}{L_{G_i}L_{F_i}^{m-i-1}p_i},
\end{equation}
if $L_{G_i}L_{F_i}^{m-i-1}p_i > 0$. Therefore, $\nu_i$ is unbounded in the (positive) infinity side. The adaptivity of an AdaCBF depends on the auxiliary dynamics (\ref{eqn:assist}) (i.e., $\bm u$ depends on $\nu_i, \forall i\in\{1,2,\dots,m-1\}$). If $\nu_i$ changes too fast, it can affect the smoothness of the control $\bm u$ obtained through solving the QPs, and thus, may damage the performance of system (\ref{eqn:affine}).

If we add control bounds on $\nu_i$, the problem feasibility may be decreased (i.e., the adaptivity of an AdaCBF is weakened). If we try to minimize each $\nu_i^2$ in the cost, $p_i$ may stay at a large value, which contradicts with the penalty method from \cite{Xiao2019} (i.e., we wish to have small enough value of $p_i$ to improve the problem feasibility). Therefore, we wish to stabilize the values of $p_i(t), \forall i\in\{1,2,\dots, m\}$ and always try to decrease $p_i$ when it is large. We usually stabilize each $p_i(t)$ to a small enough value $p_i^* > 0$ (recommended by the penalty method from \cite{Xiao2019} or by the optimal penalties learned in \cite{Xiao2019LCSS}). We choose smaller $p_i^*$ if $\alpha_i(\cdot)$ is a high order function (such as polynomial function) than a low order one. Suppose the auxiliary dynamics (\ref{eqn:assist}) are input-output linearized (otherwise, we can do input-output linearization as we have assumed that (\ref{eqn:assist}) is input-output linearizable), we can use either the tracking control from \cite{Khalil2002} or the CLF method to stabilize $p_i(t)$, i.e., if $m = 1$, we wish to minimize $(p_1 - p_1^*)^2$ (take $p_1$ as a decision variable), if $m = 2$, we define a CLF $V_1(p_1) := (p_1 - p_1^*)^2$ as in Def. \ref{def:clf}, and if $m > 2$, we find a desired state feedback $\hat p_{i,m-i}$ for $p_{i,m-i}$ in the form:
\begin{equation}\label{eqn:pdesired}
\hat p_{i,m-i}= 
\left\{
\begin{array}{lcl}
-k_1(p_{i} - p_i^*), \qquad\qquad\qquad\;\;\; i = m-2,\\
\begin{aligned}
\!-k_1\!(p_i \!-\! p_i^*) \!-\! k_2p_{i,2} \!-\!\dots, \!-\!k_{m\!-\!i\!-\!1}p_{i,m\!-\!i\!-\!1},& \\ i < m-2,&\end{aligned}
\end{array}
\right.
\end{equation}
where $k_1 > 0, k_2 > 0,\dots, k_{m-i-1} > 0$. In the last equation, if $i = m-1$, we can directly define a CLF $V_i(\bm p_i):= (p_i - p_i^*)^2$ as in Def. \ref{def:clf}.

Then we can define a CLF $V_i(\bm p_i):= (p_{i,m-i} - \hat p_{i,m-i})^2, i\in\{1,2,\dots, m-1\}$ (the relative degree of $V_i(\bm p_i)$ is one) to stabilize each $p_i$ with $c_2 = \epsilon > 0$ in Def. \ref{def:clf}, any control input $\nu_i$ should satisfy:
\begin{equation}\label{eqn:oc_clf}
 L_{F_i}V_i(\bm p_i)+L_{G_i}V_i(\bm p_i) \nu_i + \epsilon V_i(\bm p_i) \leq \delta_i
\end{equation}
where $\delta_i$ denotes a relaxation variable that we want to minimize.

In all cases, we may also want to stabilize $p_{m}$ that is not differentiated, we can just minimize $(p_m - p_m^*)^2$ (take $p_m$ as a decision variable). As discussed in the above, we wish to decrease $p_i$ when it is large, so we always want to minimize $\nu_i$. Therefore, we can reformulate the cost (\ref{eqn:cost}) by the AdaCBF in the form (let $\bm \delta := (\delta_1, \delta_2,\dots, \delta_{m-1})$):
\begin{equation}\label{eqn:cost_reform}
\begin{aligned}
J(\bm u(t), \bm \nu(t),\bm\delta(t), p_m(t)) = \int_{0}^{T}\mathcal{C}(||\bm u(t)||) + \sum_{i=1}^{m-1}W_i\nu_i(t) \\+ \sum_{i = 1}^{m-1}P_i\delta_i^2(t) + Q(p_m(t) - p_m^*)^2dt
\end{aligned}
\end{equation}
subject to (\ref{eqn:affine}), (\ref{eqn:controlvar}), (\ref{eqn:acbf}), (\ref{eqn:assist}), the HOCBF constraint in (\ref{eqn:cbf}) for each $p_i\geq 0, i\in\{1,2,\dots, m-1\}$, $p_m(t) \geq 0$, and the CLF constraint (\ref{eqn:oc_clf}). $W_i > 0, P_i > 0, i\in\{1,2,\dots, m-1\}$, and $Q \geq 0$.

Then we can use the QP based approach introduced at the end of Sec. \ref{sec:pre} to solve (\ref{eqn:cost_reform}).

$\textbf{Complexity:}$
The time complexity of QP (active-set method) is polynomial in the dimension of decision variables on average. In general, the complexity is $O(n^3)$, where $n$ denotes the dimension of the decision variable space. In the HOCBF based QP, the complexity is $O(q^3)$ (recall that $q$ is the dimension of the control $\bm u$). However, in (\ref{eqn:cost_reform}), the complexity becomes $O((q+2m-1)^3)$. We increase the adaptivity of the CBF method at the cost of more computation time, but the AdaCBF based QP is still fast enough in most problems, as we will see in the simulations.

\section{ACC PROBLEM FORMULATION}
\label{sec:ACC}

In this section, we consider a more realistic version of the adaptive cruise control (ACC) problem introduced in 
Sec.\ref{sec:hocbf:moti}, which was referred to as the simplified adaptive cruise control (SACC) problem. we consider that the safety constraint is critical and study the adaptivity of AdaCBF discussed in the last section.

 \textbf{Vehicle Dynamics:} Instead of using the simple dynamics in (\ref{eqn:simpledynamics}), we consider more accurate vehicle dynamics in the form:
\begin{equation}\label{eqn:vehicle}
\underbrace{\left[\begin{array}{c} 
	\dot x(t)\\
	\dot v(t)
	\end{array} \right]}_{ \dot {\bm x}(t)}=
\underbrace{\left[\begin{array}{c}  
	v(t)\\
	-\frac{1}{M}F_{r}(v(t))
	\end{array} \right]}_{f(\bm x(t))} + 
\underbrace{\left[\begin{array}{c}  
	0\\
	\frac{1}{M}
	\end{array} \right]}_{g(\bm x(t))}u(t)
\end{equation}
where $M$ denotes the mass of the controlled vehicle. $F_{r}(v(t))$ denotes the resistance force, which is expressed \cite{Khalil2002} as:
\begin{equation}\label{eqn:resistence}
F_{r}(v(t)) = f_0sgn(v(t)) + f_1v(t) + f_2 v^2(t),
\end{equation}
where $f_0 > 0, f_1 > 0$ and $f_2 > 0$ are scalars determined empirically. The first term in $F_{r}(v(t))$ denotes the coulomb friction force, the second term denotes the viscous friction force and the last term denotes the aerodynamic drag.

\textbf{Constraint 1} (Vehicle limitations): There are constraints
on the speed and acceleration, i.e.,
\begin{equation}\label{eqn:limitation}%
\begin{aligned} v_{min} \leq v(t)\leq v_{max}, \forall t\in[0,t_f],\\ -c_d(t)Mg\leq u(t)\leq c_a(t)Mg, \forall t\in[0,t_f], \end{aligned} 
\end{equation}
where $v_{max}>0$ and $v_{min}\geq 0$ denote the maximum and minimum allowed speeds, while $c_d(t)>0$ and $c_a(t) > 0$ are deceleration and
acceleration coefficients, respectively, and $g$ is the gravity constant.

\vspace{1mm}
\textbf{Constraint 2} (Safety): Eqn. (\ref{eqn:safety}).

\vspace{1mm}
\textbf{Objective 1} (Desired Speed): The controlled vehicle always attempts to achieve a desired speed $v_d > 0$.

\vspace{1mm}
\textbf{Objective2} (Minimum Energy Consumption): We also want to minimize the energy consumption:
\begin{equation}\label{eqn:energy}
J(u(t))=\int_{0}^{T}\left(\frac{u(t) - F_{r}(v(t))}{M}\right)^2dt,
\end{equation}

\begin{problem}\label{problem1}
	Determine control laws to achieve Objectives 1, 2 subject to Constraints 1, 2, for the controlled vehicle governed by dynamics (\ref{eqn:vehicle}).
\end{problem}
\vspace{2ex}

We use an AdaCBF to implement Constraint 2 (\ref{eqn:safety}), and use HOCBFs to impose constraint 1 (\ref{eqn:limitation}) on control input and a control Lyapunov function in Def. \ref{def:clf} to achieve Objective 1. We capture Objective 2 in the cost of the optimization problem.

\section{ACC PROBLEM REFORMULATION}
\label{sec:reform}

For {\it\textbf{Problem} 1}, we use the quadratic program (QP) - based method introduced in \cite{Aaron2014}. The relative degree of (\ref{eqn:safety}) is 2, so we defined an AdaCBF with $m = 2$.  We consider a quadratic class $\mathcal{K}$ function for $\alpha_1(\cdot)$ and a linear class $\mathcal{K}$ function for $\alpha_2(\cdot)$ in the definition of the AdaCBF.

\subsection{Desired Speed (Objective 1)} We use a  control Lyapunov function to stabilize $v(t)$ to $v_d$ and relax the corresponding constraint (\ref{eqn:clf}) to make it a soft constraint \cite{Aaron2012}. Consider a Lyapunov function $V_{acc}(\bm x(t)):= (v(t) - v_d)^2$, with $c_1 = c_2 = 1$ and $c_3 = \epsilon > 0$ in Def. \ref{def:clf}. 
Any control input $u(t)$ should satisfy
\begin{equation}\label{eqn:clfconstraint}%
\begin{aligned}
\underbrace{-\frac{2(v(t) - v_d)}{M}F_r(v(t))}_{L_fV_{acc}(\bm x(t))} + \underbrace{\epsilon(v(t) - v_d)^2}_{\epsilon V_{acc}(\bm x(t))} \\+ \underbrace{\frac{2(v(t) - v_d)}{M}}_{L_gV_{acc}(\bm x(t))}u(t) \leq \delta_{acc}(t)
\end{aligned}
\end{equation}
$\forall t\in[0,t_f]$. Here $\delta_{acc}(t)$ denotes a relaxation variable that makes (\ref{eqn:clfconstraint}) a soft constraint.

\subsection{Vehicle Limitations (Constraint 1)} The relative degrees of speed limitations are 1, we use HOCBFs with $m = 1$ to map the limitations from speed $v(t)$ to control input $u(t)$. Let $b_{1}(\bm x(t)) := v_{max} - v(t)$, $b_{2}(\bm x(t)) := v(t) - v_{min}$ and choose $\alpha_1(b_{i}) = b_{i}, i\in\{1,2\}$ in Def. \ref{def:hocbf} for both HOCBFs. Then any control input $u(t)$ should satisfy
\begin{small}
	\begin{equation}\label{eqn:MaxSpeedConstraint}%
	\begin{aligned}
	\underbrace{\frac{F_r(v(t))}{M}}_{L_fb_{1}(\bm x(t))}\! +\! \underbrace{\frac{-1}{M}}_{L_gb_{1}(\bm x(t))}u(t)\! + \underbrace{v_{max} - v(t)}_{b_{1}(\bm x(t))} \geq\! 0,
	\end{aligned}
	\end{equation}
\end{small}
\begin{small}
	\begin{equation}\label{eqn:MinSpeedConstraint}%
	\begin{aligned}
	\underbrace{\frac{-F_r(v(t))}{M}}_{L_fb_{2}(\bm x(t))}\! +\! \underbrace{\frac{1}{M}}_{L_gb_{2}(\bm x(t))}u(t)\! + \underbrace{v(t) - v_{min}}_{b_{2}(\bm x(t))} \geq\! 0.
	\end{aligned}
	\end{equation}
\end{small}

Since the control limitations are already constraints on control input, we do not need HOCBFs for them.

\subsection{Safety Constraint (Constraint 2)} Since the HOCBF constraint for (\ref{eqn:safety}) can easily conflict with the control limitations in (\ref{eqn:limitation}), we use an AdaCBF with $m = 2$ (the relative degree of the safety constraint (\ref{eqn:safety}) is two). Let $b(\bm x(t)) := x_{p}(t) - x(t) - \delta_0$, we define $\psi_1, \psi_2$ in the form:

\begin{equation} \label{eqn:linear}
\begin{aligned}
\psi_1(\bm x(t), \bm p(t)) :=& \dot b(\bm x(t)) + p_1(t)b^2(\bm x(t))\\
\psi_2(\bm x(t), \bm p(t)) := &\dot \psi_1(\bm x(t), \bm p(t)) + p_2(t)\psi_1(\bm x(t), \bm p(t))
\end{aligned}
\end{equation}

We define an auxiliary dynamics $\dot p_1 = F_1(p_1) + G_1(p_1)\nu_1$ for $p_1(t)$ (we will not take derivatives on $p_2(t)$, so we just set $p_2(t)\geq 0,\forall t\geq 0$ as a decision variable to be determined (also time-varying)) in the form 
\begin{equation}\label{eqn:lin}
\dot p_1(t) = \nu_1(t),
\end{equation}

Combining the dynamics (\ref{eqn:vehicle}) with (\ref{eqn:linear}), any control input  $u(t)$ should satisfy (the AdaCBF constraint):
\begin{equation} \label{eqn:acbf_acc}
\begin{aligned}
\underbrace{\frac{F_r(v(t))}{M}}_{L_f^2b(\bm x(t))}\! +\! \underbrace{\frac{-1}{M}}_{L_gL_fb(\bm x(t))}u(t)\! + \underbrace{b^2(\bm x(t))}_{L_{G_1}p_1(t)\alpha_2(\psi_1)}\nu_1(t) \\+ 2p_1(t)b(\bm x(t))\dot b(\bm x(t)) +  p_2(t)\psi_1(\bm x(t), \bm p(t))\geq\! 0.
\end{aligned}
\end{equation}

Since $p_1(t)$ has to be a HOCBF ($p_1(t)\geq 0$), and its relative degree is 1 for (\ref{eqn:lin}), any control input $\nu_1(t)$ should satisfy (define $\alpha_1$ as in a linear function in Def. \ref{def:hocbf}):
\begin{equation} \label{eqn:cbf_acc}
\begin{aligned}
\underbrace{0}_{L_{F_1}p_1(t)}\! +\! \underbrace{1}_{L_{G_1}p_1(t)}\nu_1(t)\! + p_1(t)\geq\! 0,
\end{aligned}
\end{equation}
with $p_1(0) > 0$.

We wish to stabilize $p_1(t)$ to a desired $p_1^* > 0$ (usually a small number), and define a CLF $V_1(p_1(t)):=(p_1(t) - p_1^*)^2$, with $c_1 = c_2 = 1$ and $c_3 = \epsilon > 0$ in Def. \ref{def:clf}. Any control input should satisfy:
\begin{equation} \label{eqn:clf_acc}
\begin{aligned}
\underbrace{0}_{L_{F_1}V_1(p_1(t))}\! +\! \underbrace{2(p_1(t) - p_1^*)}_{L_{G_1}V_1(p_1(t))}\nu_1(t) + \epsilon V_1(p_1(t)) \leq \delta_1(t).
\end{aligned}
\end{equation}
Note that we may also wish to minimize $(p_2(t) - p_2^*)^2$, where $p_2^* > 0$.

\subsection{Reformulated ACC Problem} 

We partition the time interval $[0,T]$ 
into a set of equal time intervals $\{[0, \Delta t), [\Delta t,2\Delta t),\dots\}$, where $\Delta t > 0$. In each interval $[\omega \Delta t, (\omega+1) \Delta t)$ ($\omega = 0,1,2,\dots$), we assume the control is constant (i.e., the overall control will be piece-wise constant), and reformulate (approximately) {\it\textbf{Problem} 1} as a sequence of QPs.  Specifically, 
at $t = \omega \Delta t$ ($\omega = 0,1,2,\dots$), we solve 
\begin{equation} \label{eqn:objACC}
\bm u^*(t) = \arg\min_{\bm u(t)} \frac{1}{2}\bm u(t)^TH\bm u(t) + F^T\bm u(t)
\end{equation}
\[\begin{small}
\bm u(t)\! =\! \left[\begin{array}{c}  
\! u(t)\!\\
\!\delta_{acc}(t)\!\\
\nu_1(t)\\
\delta_1(t)\\
p_2(t)
\end{array} \right]\!,
 F\! =\!  \left[\begin{array}{c} 
\!\frac{-2F_r(v(t))}{M^2}\!\\
0\\
W_1\\
0\\
-2Qp_2^*
\end{array} \right], \end{small}
\]
\[
H\! =\! \left[\begin{array}{ccccc} 
\frac{2}{M^2} & 0& 0& 0& 0\\
0 & 2p_{acc} & 0& 0& 0\\
0 & 0 & 0& 0& 0\\
0 & 0 & 0& 2P_1& 0\\
0 & 0 & 0& 0& 2Q\\
\end{array} \right]\!,
\]
where $p_{acc} > 0, W_1 > 0, P_1 > 0, Q \geq 0$ (We also assume $F$ is a constant vector in each interval), subject to 

\[
A_{\text{speed\_clf}} \bm u(t) \leq b_{\text{speed\_clf}},
\]
\[
A_{\text{limit}} \bm u(t) \leq b_{\text{limit}},
\]
\[
A_{\text{acbf}} \bm u(t) \leq b_{\text{acbf}}, 
\]
where the constraint parameters are
$$
\begin{aligned}
A_{\text{speed\_clf}} &= [L_gV_{acc}(\bm x(t)), -1, 0, 0, 0],\\
b_{\text{speed\_clf}} &= -L_fV_{acc}(\bm x(t)) - \epsilon V_{acc}(\bm x(t)).
\end{aligned}
$$
$$
\begin{aligned}
A_{\text{limit}} &= \left[\begin{array}{ccccc} 
-L_gb_{1}(\bm x(t)), & 0& 0& 0& 0\\
-L_gb_{2}(\bm x(t)), & 0& 0& 0& 0\\
1, & 0& 0& 0& 0\\
-1, & 0& 0& 0& 0
\end{array} \right],\\
b_{\text{limit}} &= \left[\begin{array}{c}  
L_fb_{1}(\bm x(t)) + b_{1}(\bm x(t))\\
L_fb_{2}(\bm x(t)) + b_{2}(\bm x(t))\\
c_a(t)Mg\\
c_d(t)Mg
\end{array} \right].
\end{aligned}
$$
\begin{small}
$$
A_{\text{acbf}} \!=\! \left[\begin{array}{ccccc} 
\!-\!L_gL_fb(\bm x(t))\! & 0 & \!-\!L_{G_1}p_1(t)\alpha_2(\psi_1)\! & 0 & \!-\!\psi_1\!\\
\!0\! & 0& -L_{G_1}p_1(t)& 0& 0\\
\!0\! & 0& \!L_{G_1}V_1(p_1(t))\!& \!-\!1& 0
\end{array} \right],
$$
\end{small}
$$
b_{\text{acbf}} = \left[\begin{array}{c}  
L_f^2b(\bm x(t)) + 2p_{1}(t)b(\bm x(t))\dot b(\bm x(t))\\
p_1(t)\\
-\epsilon V_1(p_1(t))
\end{array} \right].
$$

 After solving (\ref{eqn:objACC}), we update (\ref{eqn:vehicle}) with $u^*(t)$, update (\ref{eqn:lin}) with $\nu_1^*(t)$, and update $p_2(t)$ with $p_2^*(t)$, $\forall t\in (\omega \Delta t, (\omega+1) \Delta t)$. 
 
\begin{remark} 
	The control bound $-c_d(t)Mg$ is usually not a constant (the same for $c_a(t)Mg$, but it does not make sense to change $c_a(t)Mg$ since the AdaCBF constraint (\ref{eqn:acbf_acc}) can only conflict with $-c_d(t)Mg$). We can also add noise to (\ref{eqn:vehicle}) as the AdaCBF constraint (\ref{eqn:acbf_acc}) is relaxed by $\nu_1(t)$, and thus is adaptive to the change of the control bound and dynamics noise.
\end{remark}

\section{IMPLEMENTATION AND RESULTS}
\label{sec:case}
In this section, we present case studies for {\it\textbf{Problem} 1} to illustrate the adaptive property of the AdaCBF described in Sec.\ref{sec:hocbf}. 

 All simulations were conducted in MATLAB. We used quadprog to solve the QPs and ode45 to integrate the dynamics. The parameters are listed in Table \ref{table:param}. If we apply HOCBF implement the safety constraint (\ref{eqn:safety}) with $p_1(t) = 0.1, p_2(t) = 1$, the QP will be infeasible after the corresponding HOCBF constraint becomes active. Therefore, we need the AdaCBF to implement this safety constraint, as shown next.
 \begin{table}
 	\caption{Simulation parameters for problem \ref{problem1}}\label{table:param}
 	\begin{center}
 		\begin{tabular}{|c||c||c|c||c||c|}
 			\hline
 			Parameter & Value & Units &Parameter & Value & Units\\
 			\hline
 			\hline
 			$v(0)$ & 20& $m/s$&	$x_p(0)\!-\!x(0)$ & 100& $m$\\
 			\hline
 			$v_{0}$ & 13.89& $m/s$ & $v_d$ & 24& $m/s$\\
 			\hline
 			$M$ & 1650& $kg$ &g & 9.81& $m/s^2$\\
 			\hline
 			$f_0$ & 0.1& $N$ &$f_1$ & 5& $Ns/m$\\
 			\hline
 			$f_2$ & 0.25& $Ns^2/m$ &$\delta_0$ & 10& $m$\\
 			\hline
 			$v_{max}$ & 30& $m/s$&	$v_{min}$ & 0& $m/s$\\
 			\hline
 			$\Delta t$ & 0.1& $s$&	$\epsilon$ & 10& unitless\\
 			\hline
 			$c_a(t)$ & TBD& unitless&$c_d(t)$ & 0.4& unitless\\
 			\hline
 			$p_{acc}$ & 1& unitless & $W_1$&2& unitless\\ 
 			\hline
 			$P_1$ & $e^{12}$& unitless & $Q$&$e^{12}$& unitless\\
 			\hline
 		\end{tabular}
 	\end{center}
 	
 \end{table}

\textbf{Adaptivity to the changing control bound $-c_d(t)Mg$:} We first studies what happens when we change the lower control bound $-c_d(t)Mg$. In each simulation trajectory, we set the lower control bound coefficient $c_d(t)$ to a different constant or to be time-varying (such as linearly decrease $c_d(t)$). In this case, we set $T = 30s, p_1(0) = p_1^* = 0.1, p_2^* = 1$. We first present a case study of linearly decreasing $c_d(t)$ (for example, due to {\it tires slipping}), as shown in Fig. \ref{fig:control_cmp_1}. When we decrease $c_d(t)$ (weaken the braking capability of the vehicle) after the HOCBF constraint becomes active, the QPs can easily become infeasible in the HOCBF method, as the red line shown in Fig. \ref{fig:control_cmp_1}. In the AdaCBF method, the QPs are always feasible as shown in Fig. \ref{fig:control_cmp_1}, which shows the adaptivity of the AdaCBF to the time-varying control bound (wheels slipping). The computational time at each time step for both the HOCBF and AdaCBF methods are less than $0.01s$ in MATLAB ({\it Intel(R) Core(TM) i7-8700 CPU
@ 3.2GHz$\times 2$}). Note that there is an over-shot for the control when $b(\bm x)$ is small, we can put more weight on control $u(t)$ or decrease the weights $P_1, Q$ to alleviate this overshot after the control constraint is not active, as the light blue curve shown in  Fig. \ref{fig:control_cmp}. The simulation trajectories for different (constant) $c_d(t)$ values (for example, on {\it different road surfaces}) are shown in Figs. \ref{fig:control_cmp} and \ref{fig:adapt}.
\begin{figure}[thpb]
	\centering
	\includegraphics[scale=0.65]{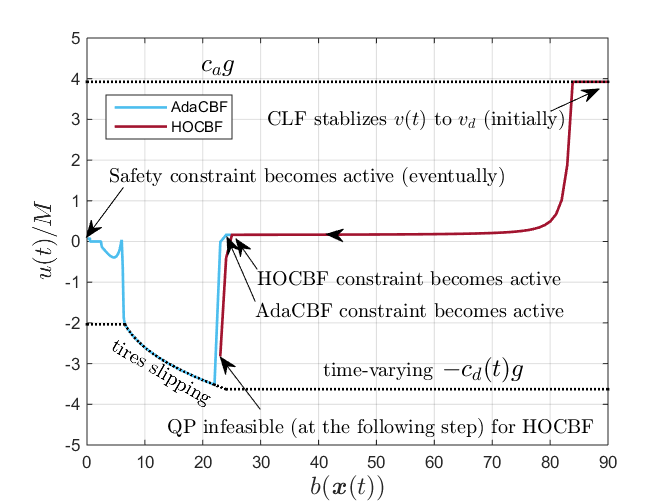}
	\caption{Control input $u(t)$ variation as $b(\bm x(t))\rightarrow 0$ for HOCBF and AdaCBF when linearly decrease $c_d(t)$ (starting from 0.37 to 0.2) after the AdaCBF (or HOCBF) constraint (\ref{eqn:acbf_acc}) becomes active. The arrow denotes the changing trend for $b(\bm x(t))$ (the AdaCBF that captures the safety constraint (\ref{eqn:safety}) (or the distance between vehicles)) with respect to time. }	
	\label{fig:control_cmp_1}
\end{figure}

\begin{figure}[thpb]
	\centering
	\includegraphics[scale=0.65]{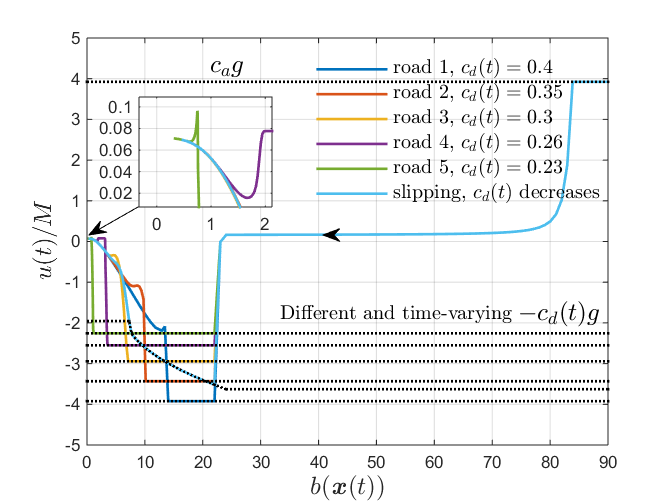}
	\caption{Control input $u(t)$ variations as $b(\bm x(t))\rightarrow 0$ under different and time-varying control lower bounds. The arrow denotes the changing trend for $b(\bm x(t))$ with respect to time. $b(\bm x)\geq 0$ implies the {\it forward invariance} of $C_1:=\{b(\bm x)\geq 0\}$}	
	\label{fig:control_cmp}
\end{figure}

\begin{figure}[thpb]
	\centering
	\includegraphics[scale=0.65]{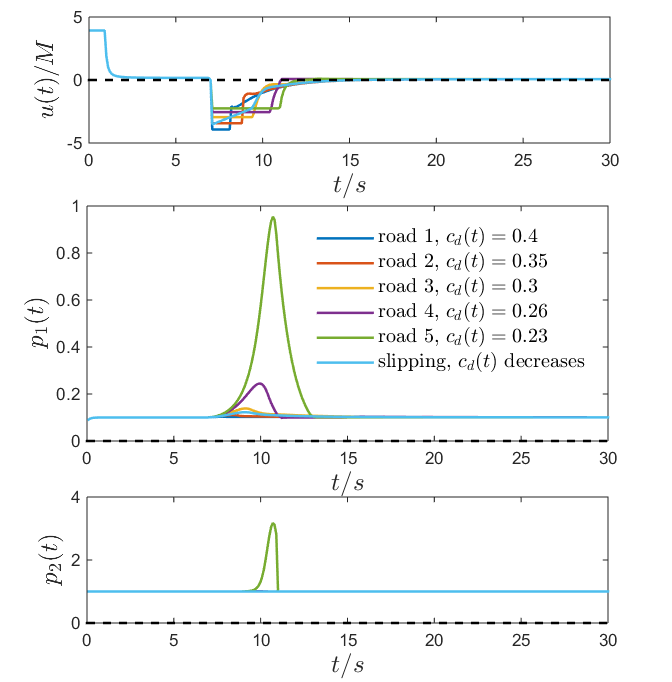}
	\caption{The penalty functions $p_1(t), p_2(t)$ and control input $u(t)$ profiles under different and time-varying control lower bounds. The value change of the penalty functions $p_1(t), p_2(t)$ demonstrates the adaptivity of the AdaCBF to the change of the control bound (or tight control bound).}
	\label{fig:adapt}
\end{figure}

As shown in Fig. \ref{fig:control_cmp}, when we set $c_d(t) = 0.4$, the QP itself has good feasibility. This induces only a little change on the penalty varible $p_1(t)$ and no change on $p_2(t)$, as shown in Fig. \ref{fig:adapt}. As we decrease $c_d(t)$ to a smaller one in another simulation (i.e., limit the braking capability of the vehicle), the variation of the penalty varibale $p_1(t)$ becomes large after the AdaCBF constraint (\ref{eqn:acbf_acc}) becomes active. As we decrease $c_d(t)$ to $0.23$, the vehicle needs to brake with $-c_d(t)Mg$ almost all the way to the safety constraint (\ref{eqn:safety}) becoming active, as green curves shown in Figs. \ref{fig:control_cmp} and \ref{fig:adapt}. This $c_d(t)$ value is close to the vehicle limit (i.e., only brake with the maximum braking force) such that the safety constraint can be satisifed. On the other hand, the penalty functions $p_1(t), p_2(t)$ both change to a big value, as shown in Fig. \ref{fig:adapt}. If we further decrease $c_d(t)$, the safety constraint (\ref{eqn:safety}) will be violated.  The value change of the penalty functions $p_1(t), p_2(t)$ demonstrates the adaptivity of the AdaCBF to the change of the control bound. The penalty method \cite{Xiao2019} shows that we wish to have smaller panelties to improve the QP feasibility before the HOCBF constraint becomes active, but the AdaCBF shows that we may actually want to increase the value of the penalties after the AdaCBF constraint becomes active, as the last frame shown in Fig. \ref{fig:adapt}. This is also demonstrated on another example shown next.

Suppose we decrese $p_1^*$ from $0.1$ to $0.02$, and we want to compare the minimum $c_d(t)$ we can take beween them, as well as study what happens when we disable the penalty function $p_2(t)$ in the AdaCBF, i.e., fix $p_2(t)$ to a constant value. Simulation results are shown in Figs. \ref{fig:control_cmp2} and \ref{fig:adapt2}.
\begin{figure}[thpb]
	\centering
	\includegraphics[scale=0.65]{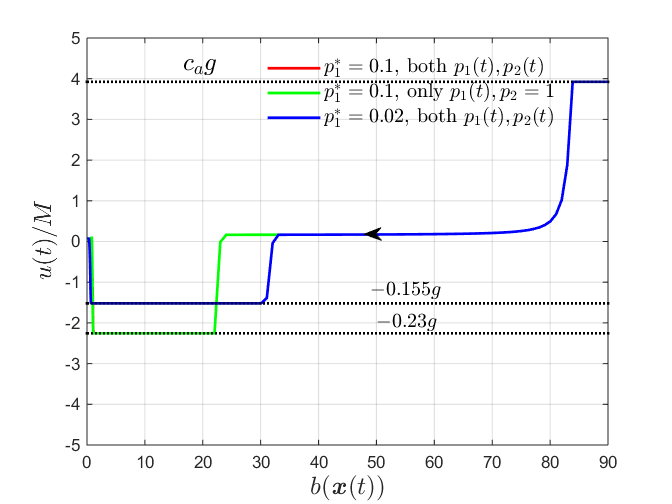}
	\caption{Control input $u(t)$ as $b(\bm x(t))\rightarrow 0$ under different $p_1^*$ values and limited penalty functions. The arrow denotes the changing trend for $b(\bm x(t))$ with respect to time.}	
	\label{fig:control_cmp2}
\end{figure}

\begin{figure}[thpb]
	\centering
	\includegraphics[scale=0.65]{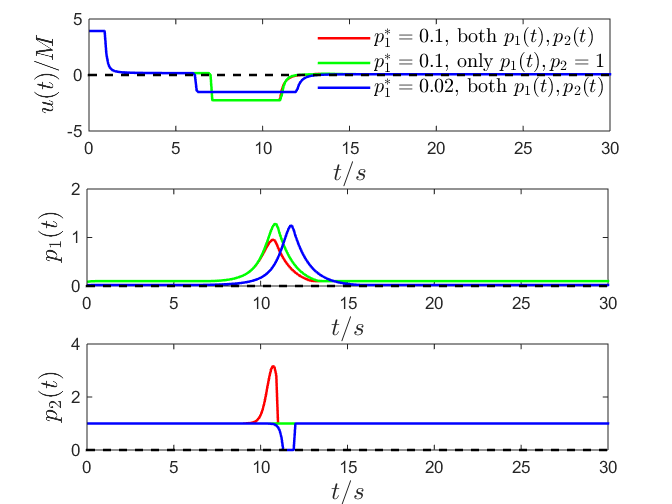}
	\caption{The penalty functions $p_1(t), p_2(t)$ and control input $u(t)$ profiles under different $p_1^*$ values and limited penalty functions. The value change of the penalty functions $p_1(t), p_2(t)$ demonstrates the adaptivity of the AdaCBF to the tight control bound.}	
	\label{fig:adapt2}
\end{figure}

We can see that when we set $p_2(t)$ to a constant instead of a penalty value, the control input profiles are almost the same, but require bigger $p_1(t)$ values after the AdaCBF constraint (\ref{eqn:acbf_acc}) becomes active. When we decrease $p_1^*$ from $0.1$ to $0.02$, we can further decrease $c_d(t)$ to 0.155, as shown in Figs. \ref{fig:control_cmp2} and \ref{fig:adapt2}. This is consistent with the penalty method \cite{Xiao2019} that we wish to take smaller penalties before the CBF constraint becomes active.

\textbf{Adaptivity to dynamics noise:} Suppose we add two noise terms $w_1(t), w_2(t)$ to the speed and acceleration in dynamics (\ref{eqn:vehicle}), respective, where $w_{1}(t),w_{2}(t)$ denote two random processes defined in an
appropriate probability space. In the simulation, $w_1(t),w_2(t)$ randomly take values in $[-2m/s, 2m/s]$ and $[-0.45m/s^2, 0.45m/s^2]$ with equal probability at time $t$, respectively. We fix the value of $c_d(t)$ to 0.23 in (\ref{eqn:limitation}) and set $T = 30s, p_1(0) = p_1^* = 0.1, p_2^* = 1$. The simulation results under different noise levels are shown in Figs. \ref{fig:adapt3} and \ref{fig:sets}, the noise is based on $[-2m/s, 2m/s]$ and $[-0.45m/s^2, 0.45m/s^2]$ for $w_1(t),w_2(t)$, respectively.

\begin{figure}[thpb]
	\centering
	\includegraphics[scale=0.65]{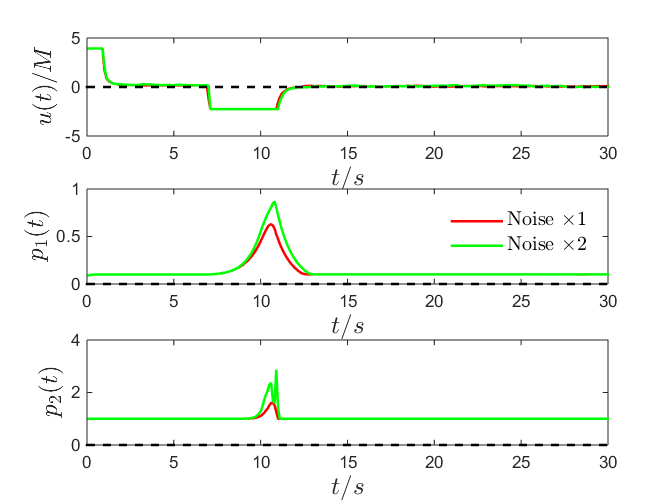}
	\caption{The penalty functions $p_1(t), p_2(t)$ and control input $u(t)$ profiles under different noise levels. The value change of the penalty functions $p_1(t), p_2(t)$ demonstrates the adaptivity of the AdaCBF to the control bound and noise.}	
	\label{fig:adapt3}
\end{figure}

\begin{figure}[thpb]
	\centering
	\includegraphics[scale=0.65]{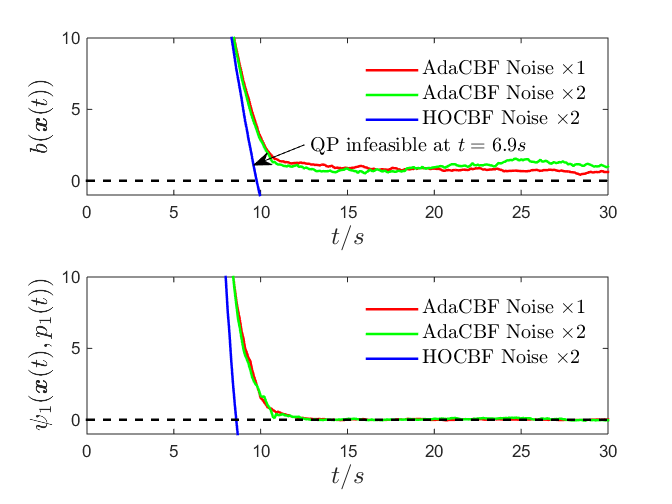}
	\caption{The profiles of $b(\bm x), \psi_1(\bm x, p_1)$ under different noise levels for AdaCBF and HOCBF, $b(\bm x)\geq 0, \psi_1(\bm x, p_1)\geq 0$ imply the {\it forward invariance} of $C_1$ and $C_2$.}	
	\label{fig:sets}
\end{figure}

 When the control constraint is active on $-c_d(t)Mg$, it can easily conflict with the HOCBF constraint (if we apply the HOCBF method) that is subjected to noise (which may make the safety constraint (\ref{eqn:safety}) violated if we apply the last moment control afterwards, as the blue line shown in Fig. \ref{fig:sets}), but the AdaCBF constraint is relaxed by the penalty functions $p_1(t)$ (through $\nu_1(t)$) and $p_2(t)$, and thus is adaptive to different dynamics noise levels and can make the QPs feasible, as shown in Fig. \ref{fig:adapt3}. The forward invariance of $C_1:=\{\bm x: b(\bm x)\geq 0\}$ and $C_2=\{(\bm x, p_1):\psi_1(\bm x, p_1)\geq 0\}$ is shown in Fig. \ref{fig:sets}. Note that $\psi_1(\bm x, p_1)$ might be temporarily negative due to noise during simulation, but will be positive again soon. This is due to the definition of $\psi_2 := \dot\psi_1 + p_2(t)\psi_1$ in (\ref{eqn:linear}). When we have $\psi_1 < 0$, the AdaCBF constraint ensures $\dot\psi_1 + p_2(t)\psi_1\geq 0$, and thus, $\dot\psi_1 \geq -p_2(t)\psi_1 > 0$ (since $p_2(t) > 0$). Therefore, $\psi_1$ will be increasing and eventually becomes positive. In this paper, we consider high order polynomial class $\mathcal{K}$ functions to make $\psi_i$ stay away from zero \cite{Xiao2019} (for example, we define $\alpha_1$ as a quadratic function in (\ref{eqn:linear})) such that $b(\bm x)\geq 0$ is guaranteed in the presence of noise. The forward invariance gurantee can also be achieved by considering the noise bounds in the AdaCBF constraint. We will compare these two approaches in future work. Note that we can also define $\alpha_2(\cdot)$ as a quadratic function in the definition of the AdaCBF in (\ref{eqn:linear}) to make $\psi_1(\bm x, p_1)$ also stay away from 0 in Fig. \ref{fig:sets}, and define $\alpha_1(\cdot)$ as a higher order polynomial function to make the AdaCBF $b(\bm x)$ stay further away to 0, and thus it can be adaptive (in the sense of both QP feasibility and forward invariance) to higher noise levels.

\section{CONCLUSION \& FUTURE WORK}
\label{sec:conclusion}

We introduce adaptive control barrier functions that
can accommodate time-varying control bounds and dynamics
noise, and also address the feasibility problem in this paper. The proposed adaptive control barrier function can also alleviate the conservativeness of the control barrier function method, and thus improve the system performance. We demonstrate the advantages of the proposed adaptive control barrier function by applying it to an adaptive cruise control problem. In the future, we will apply the adaptive control barrier function method to complex problems and systems.






\bibliographystyle{plain}
\bibliography{AdaCBF}

\end{document}